\documentclass{amsart}

\usepackage{amsmath,amsthm,amsfonts, amssymb,amscd}
\usepackage[all]{xy}

\newtheorem{theorem}{Theorem}[section]
\newtheorem{remark}[theorem]{Remark}
\newtheorem{proposition}[theorem]{Proposition}

\newcommand{\dx}{\, \mbox{\rm d}}

\newcommand{\Sh}{\mbox{\rm Sh}}

\begin{document}
\title[Smearing of Observables and Spectral Measures]{Smearing of Observables and Spectral Measures on Quantum Structures}
\author[Anatolij Dvure\v censkij]{Anatolij Dvure\v censkij$^{1,2}$}
\date{}
\maketitle
\begin{center}  \footnote{


The paper has been supported by Center of Excellence SAS -~Quantum
Technologies~-,  ERDF OP R\&D Project
meta-QUTE ITMS 26240120022, the grant VEGA No. 2/0059/12 SAV and by
CZ.1.07/2.3.00/20.0051. }
Mathematical Institute,  Slovak Academy of Sciences,\\
\v Stef\'anikova 49, SK-814 73 Bratislava, Slovakia\\
$^2$ Depar. Algebra  Geom.,  Palack\'{y} Univer.\\
CZ-771 46 Olomouc, Czech Republic\\

E-mail: {\tt
dvurecen@mat.savba.sk}
\end{center}

\begin{abstract} An observable on a quantum structure is any $\sigma$-homomorphism of quantum structures from the Borel $\sigma$-algebra of the real line into the quantum structure which is in our case a monotone $\sigma$-complete effect algebras with the Riesz Decomposition Property. We show that every observable is a smearing of a sharp observable which takes values from a Boolean $\sigma$-subalgebra of the effect algebra, and we prove that for every element of the effect algebra there is its spectral measure.
\end{abstract}

\section{Introduction}

D-posets introduced by K\^opka and Chovanec \cite{KoCh} and effect algebras introduced by Foulis and Bennet \cite{FoBe}  became the last two decades  very important quantum structures which model quantum mechanical events. Both structures are partial algebraic structures. Subtraction of two comparable events is a basic notion for D-posets, and addition of two mutually excluding events is a basic one for effect algebras. We recall that both structures are equivalent as it was mentioned in \cite{FoBe}. In our paper we will deal only with effect algebras.

We note that a prototypical example of effect algebras, important mainly  for measurements in Hilbert space quantum mechanics,  is the system $\mathcal E(H)$ of all Hermitian operators of a (real, complex or quaternionic)  Hilbert space $H$ that lie between the zero and the identity operator. The system $\mathcal E(H)$ is used for modeling unsharp observables via POV-measures (= positive operator valued measure) in measurements in  quantum mechanics.

To describe a measurement on a quantum structure, we use the notion of an observable. This is an analogue of a random variable in a classical measurement. In our case it is simply a $\sigma$-homomorphism of effect algebras from the Borel $\sigma$-algebra $\mathcal B(\mathbb R)$ of the real line $\mathbb R$ into the quantum structure.

Observables as an important tool of quantum structures are intensively studied by many authors. A functional calculus of observables on D-posets is presented in \cite{KoCh}. The series of papers \cite{Pul, JPV, JPV1} is dedicated  to observables studied on lattice effect algebras and $\sigma$-MV-algebras exhibiting spectral properties and smearing of fuzzy observables by sharp observables using a kind of a Markov kernel. In \cite{DvKu}, it was shown that in many important structures, a partial information on an observable known only on intervals of the form $(-\infty,t),$ $t \in \mathbb R,$ is sufficient to derive the whole information on the observable.

The main tool in our research will be applications of the Loomis-Sikorski Theorem for monotone $\sigma$-complete effect algebras with the Riesz Decomposition Property (RDP) proved in \cite{BCD}. This Theorem says that our structure is a $\sigma$-homomorphic image of a monotone $\sigma$-complete effect algebra of fuzzy sets where all algebraic operations are defined by points.  This generalizes analogous results proved for a special case of monotone $\sigma$-complete effect algebras, called $\sigma$-complete MV-algebras, see \cite{Mun, Dvu1}. We recall that RDP is a special type of distributivity which in our case means the possibility of performing a joint refinement of two decompositions. It has an important consequence that an effect algebra with RDP is always an interval in a partially ordered group  with strong unit (= order unit), see \cite{Rav}.

The present paper is inspired by the research in \cite{Pul}. We have two aims. First,  we show that every observable of a monotone $\sigma$-complete effect algebra $M$ with RDP is a smearing of a sharp observable, where a sharp observable means that its values are in the biggest Boolean $\sigma$-subalgebra of $M.$ Second, we show that every element $a$ of $M$ admits a spectral measure $\Lambda_a,$  which is a sharp observable concentrated on the real interval $[0,1].$ In the language of spectral theory of self-adjoint operators, $\Lambda_a$ is the spectral measure of the element $a.$ Analogous questions were inspected also in \cite{Pul} for $\sigma$-complete MV-algebras and unital Dedekind $\sigma$-complete $\ell$-groups \cite{Pul1}.

The paper is organized as follows. Section 2 is gathering necessary notions on effect algebras. Section 3 is studying a canonical representation as well as  regular representations of monotone $\sigma$-complete effect algebras with RDP. Finally, Section 4 presents the main results on smearing of observables by sharp ones, and spectral measures of elements are established.

\section{Basic Notions of Effect Algebras}

We recall that according to \cite{FoBe}, an {\it effect algebra} is  a partial algebra $M =
(M;+,0,1)$ with a partially defined operation $+$ and two constant
elements $0$ and $1$  such that, for all $a,b,c \in M$,
\begin{enumerate}

\item[(i)] $a+b$ is defined in $M$ if and only if $b+a$ is defined, and in
such a case $a+b = b+a;$

 \item[(ii)] $a+b$ and $(a+b)+c$ are defined if and
only if $b+c$ and $a+(b+c)$ are defined, and in such a case $(a+b)+c
= a+(b+c);$

 \item[(iii)] for any $a \in M$, there exists a unique
element $a' \in M$ such that $a+a'=1;$

 \item[(iv)] if $a+1$ is defined in $M$, then $a=0.$
\end{enumerate}

If we define $a \le b$ if and only if there exists an element $c \in
M$ such that $a+c = b$, then $\le$ is a partial ordering on $M$, and
we write $c:=b-a;$ then $a' = 1 - a$ for any $a \in M.$ As a basic source of information about effect algebras we can recommend the monograph \cite{DvPu}. An effect algebra is not necessarily a lattice. We recall that a {\it homomorphism} is any mapping of two effect algebras which preserves $1$ and the addition $+.$

We show two kinds of important effect algebras. (1) If $M$ is a system of fuzzy sets on $\Omega,$ that is $M \subseteq [0,1]^\Omega,$ such that
(i) $1 \in M$, (ii) $f \in M$ implies $1-f \in M$, and (iii) if $f,g
\in M$ and $f(\omega) \le 1 -g(\omega)$ for any $\omega \in \Omega$,
then $f+g \in M,$ then $M$ is an effect algebra of fuzzy sets which
is not necessarily a Boolean algebra as well as not a lattice. (2) If $G$ is a partially ordered group written additively, $u \in G^+$, then $\Gamma(G,u):=[0,u]=\{g \in G: 0 \le g \le u\}$ is an effect algebra with $0=0,$ $1=u$ and $+$ is the group addition of elements if it exists in $\Gamma(G,u).$

We say that an effect algebra $M$ satisfies the Riesz Decomposition Property (RDP for short) if for all $a_1,a_2,b_1,b_2 \in M$ such that $a_1 + a_2 = b_1+b_2,$ there are four elements $c_{11},c_{12},c_{21},c_{22}$ such that $a_1 = c_{11}+c_{12},$ $a_2= c_{21}+c_{22},$ $b_1= c_{11} + c_{21}$ and $b_2= c_{12}+c_{22}.$

We note that an element of an effect algebra $M$ is said to be {\it sharp} if $a \wedge a'$ exists in $M$ and $a\wedge a'=0.$ Let $\Sh(M)$ be the set of sharp elements of $M.$ Then (i) $0,1\in \Sh(M),$ (ii) if $a \in \Sh(M),$ then $a'\in \Sh(M).$ If $M$ is a lattice effect algebra, then $\Sh(M)$ is an orthomodular lattice which is a subalgebra and a sublattice of $M,$  \cite{JeRi}. If an effect algebra $M$ satisfies RDP, then by \cite[Thm 3.2]{Dvu2}, $\Sh(M)$ is even a Boolean algebra, and an element $a$ is sharp iff $a\wedge a'$ is defined in $M$ and $a\wedge a'=0.$

An effect algebra $M$ is {\it monotone} $\sigma$-{\it complete} if, for any sequence $a_1 \le a_2\le \cdots,$ the element $a = \bigvee_n a_n$  is defined in $M$ (we write $\{a_n\}\nearrow a$). We recall that a mapping $x: \mathcal B(\mathbb R) \to M$ is said to be an {\it observable} on $M$ if (i) $x(\mathbb R)=1,$ (ii) if $E$ and $F$ are mutually disjoint Borel sets, then $x(E \cup F)=x(E)+x(F),$ where $+$ is the partial addition on $M,$ and (iii) if $\{E_i\}$ is a sequence of Borel sets such that $E_i \subseteq E_{i+1}$ for every $i$ and $E= \bigcup_i E_i,$ then $x(E) = \bigvee_i x(E_i).$ In other words, an observable is a $\sigma$-homomorphism of effect algebras.

An {\it effect-tribe}  is any system ${\mathcal T}$ of fuzzy sets on
$\Omega\ne \emptyset $ such that (i) $1 \in {\mathcal T}$, (ii) if $f
\in {\mathcal T},$ then $1-f \in {\mathcal T}$, (iii) if $f,g \in {\mathcal T}$,
$f \le 1-g$, then $f+g \in {\mathcal T},$ and (iv) for any sequence
$\{f_n\}$ of elements of ${\mathcal T}$ such that $f_n \nearrow f$
(pointwise), then $f \in {\mathcal T}$. It is evident that any
effect-tribe is a monotone $\sigma$-complete effect algebra. We recall that e.g. $\mathcal E(H)$ can be represented as an effect-tribe, but RDP fails for it.

A very important subclass of effect algebras is the class of
MV-algebras introduced by Chang \cite{Cha}.

We recall that an MV-algebra is an algebra $M = (M;\oplus, ^*,0,1)$
of type (2,1,0,0) such that, for all $a,b,c \in M$, we have

\begin{enumerate}
\item[(i)]  $a \oplus  b = b \oplus a$;
\item[(ii)] $(a\oplus b)\oplus c = a \oplus (b \oplus c)$;
\item[(iii)] $a\oplus 0 = a;$
\item[(iv)] $a\oplus 1= 1;$
\item[(v)] $(a^*)^* = a;$
\item[(vi)] $a\oplus a^* =1;$
\item[(vii)] $0^* = 1;$
\item[(viii)] $(a^*\oplus b)^*\oplus b=(a\oplus b^*)^*\oplus a.$
\end{enumerate}

If we define a partial operation $+$ on $M$ in such a way that $a+b$
is defined in $M$ if and only if $a \le b^*$ and  we set
$a+b:=a\oplus b$, then $(M;+,0,1)$ is an effect algebra with RDP which is a distributive lattice.

We recall that a {\it tribe} on $\Omega \ne \emptyset$
is a collection ${\mathcal T}$ of fuzzy sets from $[0,1]^\Omega$ such
that (i) $1 \in {\mathcal T}$, (ii) if $f \in {\mathcal T}$, then $1 - f \in
{\mathcal T},$ and (iii) if $\{f_n\}$ is a sequence from ${\mathcal T}$,
then $\min \{\sum_{n=1}^\infty f_n,1 \}\in {\mathcal T}.$  A tribe is
always a $\sigma$-complete MV-algebra of fuzzy sets where MV-operations are defined by points.

\section{Loomis--Sikorski Theorem}

In this section, we study representations of $\sigma$-complete effect algebra with RDP: a canonical representation and regular ones.
A basic tool of investigation in our study is an application of the Loomis-Sikorski Theorem of monotone $\sigma$-complete effect algebras with RDP proved in  \cite{BCD}:

\begin{theorem}\label{LoSiEA}
Every monotone $\sigma$-complete effect algebra with RDP is a $\sigma$-epi\-morphic image of an effect-tribe with RDP.
\end{theorem}

For $\sigma$-complete MV-algebras, we have a Loomis-Sikorski type representation which was proved independently in  \cite{Mun, Dvu1}:

\begin{theorem}\label{LoSiMV}
Every  $\sigma$-complete MV-algebra is a $\sigma$-epimorphic image of a tribe.
\end{theorem}

The proofs of these results  proved in \cite{BCD,Dvu1} used the notion of states, analogues  of probability measures.

We recall that a {\it state} on an effect algebra $M$ is any mapping $s: M \to [0,1]$ such that (i) $s(1) =1$ and (ii) $s(a+b) = s(a) + s(b)$ whenever $a+b$ is defined in $M.$  We denote by ${\mathcal S}(M)$ the
set of all states on $M$. It can happen that ${\mathcal S}(M)$ is empty,
see e.g. \cite[Ex 4.2.4]{DvPu}. But if
$M$ satisfies RDP, ${\mathcal S}(M)$ is nonempty, see \cite{Rav} and \cite[Cor. 4.4]{Goo}.  In particular, every MV-algebra admits a state. We recall that  ${\mathcal S}(M)$ is always a convex set. A state $s$ is said to be {\it extremal} if $s = \lambda s_1 +
(1-\lambda)s_2$ for $\lambda \in (0,1)$ implies $s = s_1 = s_2.$
We denote by $\partial_e \mathcal S(M)$  the set of all extremal states of
${\mathcal S}(M)$. We say that a net of states, $\{s_\alpha\}$, on $M$
{\it weakly converges} to a state $s$ on $M$ if $s_\alpha(a) \to
s(a)$ for any $a \in M$. In this topology,  ${\mathcal S}(M)$ is a
compact Hausdorff topological space and every state on $M$ lies in
the weak closure of the convex hull of the extremal states as it
follows from the Krein-Mil'man theorem. Hence, $\mathcal S(M)$ is empty iff so is $\partial_e \mathcal S(M).$

If $\mathcal S(M)$ is non-void, given an element $a \in M,$ we define a function $\hat a: \mathcal S(M) \to [0,1]$ by
$$
\hat a(s):= s(a), \ s \in \mathcal S(M).
$$
Then $\hat a$ is a continuous affine function on $\mathcal S(M).$

It is important to note that if $M$ is an MV-algebra, $\partial_e \mathcal S(M)$ is always a compact set. In general, this is not true for every  effect algebra. However,  a delicate  result of Choquet \cite[page 49]{Alf}  says that  the set of extremal states  is
always a Baire space, i.e. the Baire Category Theorem holds for $\partial_e {\mathcal
S}(M).$

Let $f$ be a real-valued function on ${\mathcal S}(M).$  We define
$$
N(f) :=\{s \in \partial_e{\mathcal S}(M) :\, f(s) \ne 0\}.\eqno(3.1)
$$

The proof of the Loomis-Sikorski Theorem from \cite[Thm 4.1]{BCD} used an effect-tribe $\mathcal T$ of fuzzy sets defined on $\Omega:= \mathcal S(M),$ and the effect-tribe $\mathcal T$ is the class of all fuzzy sets $f$ on
${\mathcal S}(M)$  with the property that there exists $b \in M$ such that
$N(f-\hat b)$ is a meager subset of $\partial_e {\mathcal S}(M)$ (in the
relative topology), then we write $f \sim b$. The $\sigma$-homomorphism $h$ was then defined by $h(f):=b$ if $f \sim b.$ We call this triple $(\Omega,\mathcal T,h)$ the {\it canonical representation} of $M.$ Every triple $(\Omega,\mathcal T,h)$ such that $\mathcal T$ is an effect-tribe of fuzzy sets on $\Omega$ and $h$ maps $\sigma$-homomorphically  $\mathcal T$ onto $M$ is said to be a {\it representation} of $M.$

\begin{proposition}\label{pr:3.3}
Let $(\Omega,\mathcal T, h)$ be the canonical representation of a monotone $\sigma$-complete effect algebra $M$ with RDP.
\begin{enumerate}
\item[{\rm (ii)}]
If $a\le b,$ $a,b \in M$, there are $f,g \in \mathcal T$ such that $f\le g$ and $h(f)= a,$ $h(g)=b.$

\item[{\rm (ii)}]
If $f,g \in \mathcal T,$ $f\le g,$ and let $c$ be an element of $M$ such that $h(f)\le c\le h(g).$  Then there exists a function $s \in \mathcal T$ such that $f\le s\le g$ and $h(s)=c.$

\end{enumerate}
\end{proposition}

\begin{proof}

(i) Let $f \sim a$ and $g\sim b$ for some $f,g \in \mathcal T.$ We have  $\max\{f,g\} \sim a\vee b =b,$ which entails $\max\{f,g\} \in \mathcal T$ and $\max\{f,g\} \sim b.$ In a similar way, $\min\{f,g\}\in \mathcal T$ and $\min\{f,g\} \sim a.$

(ii) Since $h$ is surjective, there is a function  $s_1\in \mathcal T$ such that $h(s_1)=c.$ If we set $s = \max\{f, \min\{g,s_1\}\},$ by (1), $s \in \mathcal T $ and    $s$ is the function in question.
\end{proof}

Let $\mathcal T$ be an effect-tribe on $\Omega$ and let $\mathcal B_0(M):=\{A\subseteq \Omega: \chi_A \in \Sh(\mathcal T)\}.$ By \cite[Prop 4.2]{Dvu3}, $\mathcal B_0(\mathcal T)$ is a $\sigma$-algebra of subsets of $\Omega.$ Let $\mathcal S_0(\mathcal T):=\{A \subseteq \Omega: \chi_A \in \mathcal T\}.$ If $\mathcal T$ satisfies RDP, then $\mathcal B_0(\mathcal T)
= \mathcal S_0(\mathcal T),$ \cite[p. 72]{Dvu3}. By \cite[Ex. 4.3]{Dvu3}, there is an effect-tribe $\mathcal T$ with RDP such that not every $f \in \mathcal T$ is $\mathcal B_0(\mathcal T)$-measurable. However, if $\mathcal T$ is a tribe, then every $f \in \mathcal T$ is $\mathcal B_0(\mathcal T)$-measurable, \cite{BuKl}. On the other hand, by \cite[Prop 4.7]{Dvu3}  if an effect-tribe $\mathcal T$ satisfies RDP, then $\mathcal T'$, the system of all functions $f \in \mathcal T$ such that $f$ is $\mathcal B_0(\mathcal T)$-measurable, is an effect-tribe and $\mathcal B_0(\mathcal T) = \mathcal B_0(\mathcal T').$

\begin{theorem}\label{th:5.1}
The canonical representation of a monotone $\sigma$-complete effect algebra $M$ with RDP has the property $h(f)=0$ if and only if $\chi_{N(f)} \in \mathcal T$ and $h(\chi_{N(f)})=0.$ In addition, $h$ maps $\mathcal B_0(\mathcal T)$ onto $\Sh(M).$
\end{theorem}

\begin{proof}
Let $(\Omega,\mathcal T,h)$ be the canonical representation of $M$ used from \cite[Thm 4.1]{BCD}. We have that $\mathcal T$ consists of all functions $f \in [0,1]^\Omega$ such that $f \sim b$ for some $b \in M;$ where we have $\Omega =\mathcal S(M)$ and $\Omega_0 =\partial_e \mathcal S(E).$
Assume that for $f \in \mathcal T,$ we have $h(f)=0.$ This means that $f \sim 0$, that is, $N(f)$ is a meager set. Hence, $\{\omega\in \Omega_0: \chi_{N(f)}(\omega) \ne 0\}= \{\omega\in \Omega_0: f(\omega)\ne 0\}$ is meager. Whence, $\chi_{N(f)} \in \mathcal T$ and $\chi_{N(f)} \sim 0$ and $h(\chi_{N(f)})=0.$

Conversely, let for some $f \in \mathcal T,$ $\chi_{N(f)} \in \mathcal T,$ and $h(\chi_{N(f)})=0.$ That is, $\chi_{N(f)} \sim 0,$ and $N(f)$ is meager which entails $f \sim 0.$

Let $f=\chi_A \in \mathcal B_0(\mathcal T)$  and let $g \in M$ be such that $g \le h(f)$ and $g\le h(f').$ Assume that $g_1 \in \mathcal T$ be such that $h(g_1)=g.$ By (2) of Proposition  \ref{pr:3.3}, the functions $g_2:= \min\{f,g_1\}$ and $g_3:= \min\{f',g_1\}$ belong to $\mathcal T,$ and $h(g_2)=g=h(g_3).$  Again by (2) of Proposition  \ref{pr:3.3}, the function $g_4=\min\{g_2,g_3\}\in \mathcal T$ and $h(g_4)=g.$ But $g_4 \le f$ and $g_4\le 1-f$ so that $g_4=0$ and $h(g_4)=g=0,$ and $h(f) \in \Sh(M).$

Now let $b \in \Sh(M).$ Then $s(b)\in \{0,1\}$ for any extremal state $s$ on $M.$ Define a function $f_b$ on $\Omega$ by $f_b(s)=s(b)$ if $s$ is an extremal state, otherwise, $f_b(s):=0.$  Then $f_b \sim b,$  $f_b \in \mathcal B_0(\mathcal T),$ and $h(f_b)=0.$
\end{proof}

The last result can be generalized as follows.

Let $\Omega_0$ be a subset of a set $\Omega \ne \emptyset.$ We recall that a $\sigma$-{\it ideal} on $\Omega_0$ is a non-empty system $\mathcal I$  of subsets of $\Omega_0$ such that (i) if $A \subseteq B \in \mathcal I,$ then $A \in \mathcal I,$ and (ii) if $A_n \in \mathcal I,$ $n \ge 1,$ then $\bigcup_n A_n \in \mathcal I.$ For example, let $M$ be an effect algebra  and let $\Omega= \mathcal S(M)$ and $\Omega_0=\partial_e \mathcal S(M).$  The set of all meagre subsets of $\Omega_0$ is a $\sigma$-ideal.

Let $f$ be a real-valued function on $\Omega.$ We define
$$
N_{\Omega_0}(f):= \{\omega \in \Omega_0: f(\omega)\ne 0\}.
$$

Let $(\Omega,\mathcal T,h)$ be a representation  of a monotone $\sigma$-complete effect algebra $M.$ We say that $(\Omega,\mathcal T,h)$ is {\it regular} if $h(f)=0$ iff $\chi_{N_{\Omega_0}(f)} \in \mathcal T$ and $h(\chi_{N_{\Omega_0}(f)})=0.$

\begin{theorem}\label{th:5.2}
Let $(\Omega,\mathcal T, h)$ be a representation of a monotone $\sigma$-complete effect algebra $M$ with RDP and let $\mathcal T$ have RDP. Let $\mathcal I_{\Omega_0}$ be an ideal of subsets of a fixed subset $\Omega_0$ of $\Omega$ such that $f \in [0,1]^\Omega$ belongs to $\mathcal T$ if and only if  there exists a function $g \in \mathcal T$ such that $N_{\Omega_0}(f-g) \in \mathcal I_{\Omega_0}.$ Then $(\Omega, \mathcal T, h)$ is regular and $h(f)=h(g)$ if and only if $N_{\Omega}(f-g)\in \mathcal I_{\Omega}.$

In addition, {\rm (1)}  suppose every $f\in \mathcal T$ is $\mathcal B_0(\mathcal T)$-measurable, and if $h(f)\le h(g),$ then $h(\chi_A)=0,$ where
$$
A :=\{\omega \in \Omega: f(\omega) > g(\omega)\}.
$$
Then $h(\mathcal B_0(\mathcal T))\subseteq \Sh(M).$

{\rm(2)} If $f\wedge (1-f) \in \mathcal T$ for every $f \in \mathcal T,$ then $h(\mathcal B_0(\mathcal T))= \Sh(M).$
\end{theorem}

\begin{proof}
Let $h(f)=0.$  Since $h(0)=0$ and $0 \in \mathcal T,$ we have $N_{\Omega_0}(f)= N_{\Omega_0}(f-0) \in \mathcal I_{\Omega_0}.$ Therefore, $N_{\Omega_0}(\chi_{N_{\Omega_0}(f)}-0)= N_{\Omega_0}(\chi_{N_{\Omega_0}(f)}) = N_{\Omega_0}(f) \in \mathcal I_{\Omega_0},$ which entails $\chi_{N_{\Omega_0}(f)} \in \mathcal T$ and $h(\chi_{N_{\Omega_0}(f)})=0.$

Conversely, let for  $f\in [0,1]^\Omega$ we have $\chi_{N_{\Omega_0}(f)} \in \mathcal T$ and $h(\chi_{N_{\Omega_0}(f)})=0.$ Then $N_{\Omega_0}(f) = N_{\Omega_0}(\chi_{N_{\Omega_0}(f)})   \in \mathcal I_{\Omega_0}$ which implies $f \in \mathcal T.$  Set $f_0 = \max\{f, \chi_{N_{\Omega_0}(f)}\}$. Then $f_0 = \chi_{N_{\Omega_0}(f)}$ and $N_{\Omega_0}(f_0-f)  \subseteq N_{\Omega_0}(f)$ which yields $N_{\Omega_0}(f_0-f) \in \mathcal I_{\Omega_0}.$ Hence, $f_0 - f \in \mathcal T$ and $0=h(f_0-f)=h(f_0)-h(f)=0-h(f)$ from which we get $h(f)=0,$ and $(\Omega,\mathcal T,h)$ is a regular representation.

Now let $f,g \in \mathcal T$ are such that $N_{\Omega_0}(f-g)\in \mathcal I_{\Omega_0}.$ We assert $h(f)=h(g).$ Indeed, define $g_0 =\max\{f,g\}.$  Then $N_{\Omega_0}(g_0-f) \subseteq N_{\Omega_0}(f-g)$ and whence, $N_{\Omega_0}(g_0-f)\in \mathcal I_{\Omega_0},$ which yields $g_0\in \mathcal T,$ $h(g_0-f)=0$ and $h(g_0)=h(f).$ Similarly $N_{\Omega_0}(g_0-g) \in I_{\Omega_0}$ and $h(g_0)=h(g).$ Therefore, $h(f)=h(g_0)=h(g).$

For the rest of the proof assume that every $f\in \mathcal T$ is $\mathcal B_0(\mathcal T)$-measurable and  condition (1) of our hypotheses holds.

\vspace{2mm} {\it Claim 1.} {\it If $a\le b,$ $a,b \in M$, there are $f,g \in \mathcal T$ such that $f\le g$ and $h(f)= a,$ $h(g)=b.$}
\vspace{2mm}

Let $h(f)=a,$ $h(g)=b.$  If $A=\{\omega \in \Omega: f(\omega) > g(\omega)\},$ by the assumption, $A \in \mathcal B_0(\mathcal T)$ and $h(\chi_A)=0.$ By \cite[Lem 4.1]{Dvu3}, for any $A \in \mathcal B_0(\mathcal T)$ and any $f\in \mathcal T,$  $f\chi_A = \min\{f, \chi_A\} \in \mathcal T,$ where $f\chi_A$ means the product of two functions.

Define $f_0 = \max\{f,g\}.$ If $\omega \in A,$ $f_0(\omega)= f(\omega)$ and if $\omega\in A^c,$ then $f_0(\omega)=g(\omega).$ Therefore, $N_{\Omega_0}(f_0\chi_A -f\chi_A)= \{\omega\in A\cap \Omega_0: f(\omega)>f(\omega)\} = \emptyset \in \mathcal I_{\Omega}$ and $N_{\Omega_0}(f_0\chi_{A^c} -g\chi_{A^c})= \{\omega\in A\cap \Omega_0: g(\omega)>g(\omega)\} =\emptyset \in \mathcal I_{\Omega_0}.$
Hence, $f\chi_A, f\chi_{A^c} \in \mathcal T,$  $h(f_0\chi_A) = h(f\chi_A),$ $h(f_0\chi_{A^c}) = h(g\chi_{A^c})$ and consequently, $f_0 = f_0\chi_A + f_0\chi_{A^c} \in \mathcal T,$ and $f_0 = f\chi_A + g\chi_{A^c}.$

Calculate: $h(f_0)=h(f\chi_A)+h(g\chi_{A^c}).$ But $h(f\chi_A)=h(f\wedge \chi_A) \le h(\chi_A)=0,$ and $h(g\chi_{A^c})=h(g)-h(g\chi_{A})=h(g)$ which gets $h(f_0)=h(g).$

In the same way we can show that $g_0=\min\{f,g\} \in \mathcal T,$  $h(g_0\chi_A)=h(g\chi_A),$  $h(g_0\chi_{A^c})= h(f\chi_{A^c}),$ and $h(g_0)=h(f).$

\vspace{2mm} {\it Claim 2.} {\it If $f,g \in \mathcal T,$ $f\le g,$ and let $c$ be an element of $M$ such that $h(f)\le c\le h(g).$  Then there exists a function $s \in \mathcal T$ such that $f\le s\le g$ and $h(s)=c.$}

\vspace{2mm}
Since $h$ is onto, there is a function  $s_1\in \mathcal T$ such that $h(s_1)=c.$ If we set $s = \max\{f, \min\{g,s_1\}\},$ by Claim 1, $s \in \mathcal T $ and    $s$ is the function in question.

\vspace{2mm}

Now we assume $f\in \mathcal B_0(\mathcal T),$ and let $b \in M$ be such that $b \le h(f),1-h(f).$

Choose a function $g_1 \in \mathcal T$  such that $h(g_1)=b.$ By Claim 2, the functions $g_2:= \min\{f,g_1\} \in \mathcal T$ and $g_3:= \min\{f',g_1\} \in \mathcal T,$ and $h(g_2)=g=h(g_3).$  Again applying Claim 2, the function $g_4=\min\{g_2,g_3\}\in \mathcal T$ and $h(g_4)=g.$ But $g_4 \le f$ and $g_4\le 1-f$ so that $g_4=0$ and $h(g_4)=g=0,$ and $h(f) \in \Sh(M).$

\vspace{2mm}

Finally assume (2), and let $b \in \Sh(M)$ and choose $g \in \mathcal T$  such that $h(g)=b.$ Let $f \in \mathcal T$ be any function  $f \le g, 1-g$, then $h(f)\le b, b'$ giving $h(f)=0.$ Since $h$ is regular, $\chi_{N_{\Omega_0}(f)}\in \mathcal B_0(\mathcal T)$  and $h(\chi_{N_{\Omega_0}(f)})=0.$ Since $g_0=\min\{g,1-g\}\in \mathcal T$ and $g_0\le g, 1-g,$ we have $h(\chi_{N_{\Omega_0}(g_0)})=0,$ and $N_{\Omega_0}(g_0) \in \mathcal I_{\Omega_0}.$ Set $G = \{\omega \in \Omega: g(\omega)=1\}.$ Then $N_{\Omega_0}(g-\chi_G) =\{\omega \in \Omega_0: g(\omega)\ne \chi_G(\omega)\} = \{\omega \in \Omega_0: g\ne 0\}\cap \{\omega \in \Omega_0: g(\omega)\ne 1\} = N_{\Omega_0(g_0)} \in \mathcal I_{\Omega_0}.$ This proves that $h$ maps $\mathcal B_0(\mathcal T)$ onto $\Sh(M).$
\end{proof}

We recall that in the latter theorem, the conditions are satisfied e.g. if $M$ is a $\sigma$-complete MV-algebra and $\mathcal T$ is a tribe.

\section{Smearing of Observables and Spectral Measures}

This section is the main body of the paper. It presents results concerning smearing of observables by a sharp observable and a spectral measure of a given element.

The notion of an observable can be literally extended to any $\sigma$-homomorphism of effect algebras $\xi: \mathcal S \to M,$ where $\mathcal S$ is a $\sigma$-algebra of subsets of a set $\Omega.$ An observable $\xi$ is {\it sharp} if $\xi(\mathcal S)\subseteq \Sh(M).$


We recall that a state $s$ on a monotone $\sigma$-complete effect algebra $M$ is $\sigma$-additive if $\{a_n\}\nearrow a$ implies $s(a)=\lim_n s(a_n).$ Let $\mathcal S_\sigma(M)$ denote the system of $\sigma$-additive states on $M.$ We recall that there is even a Boolean $\sigma$-algebra which has lot of states but no $\sigma$-additive state, \cite{Sik}.

\begin{theorem}\label{th:5.3}
Let $M$ be a monotone $\sigma$-complete effect algebra with RDP having at least one $\sigma$-additive state and let $(\Omega,\mathcal T,h)$ be the canonical representation of $M$ such that every $f \in \mathcal T$ is $\mathcal B_0(\mathcal T)$-measurable.  There is a sharp observable $\xi$ from $\mathcal B_0(\mathcal T)$ into $M$ such that given Given an observable $x$ on $M,$ $m \in \mathcal S_\sigma(M)$ and $E \in \mathcal B(\mathbb R)$
$$
m(x(E))= \int_\Omega f_E(\omega)\, \dx m\circ \xi(\omega), \eqno(4.1)
$$
where $f_E$ is an arbitrary function from $\mathcal T$ such that $h(f_E)=x(E).$
\end{theorem}

\begin{proof}
Let $m$ be a $\sigma$-additive state on $M,$ then $m\circ h$ is a $\sigma$-additive state on $\mathcal T.$ By the generalized theorem of Klement and Butnariu holding for effect-tribes with condition that every $f\in \mathcal T$ is $\mathcal B_0(\mathcal T)$-measurable, \cite[Thm 4.4]{Dvu2}, for every $\sigma$-additive state $s$ on $\mathcal T,$ there is a unique probability measure, $P_s,$ on $\mathcal B_0(\mathcal T)$ such that
$$ s(f)= \int_\Omega f(\omega)\, \dx P_s(\omega), \quad f \in \mathcal T. \eqno(4.2)
$$

Let $x:\mathcal B_0(\mathbb R)\to M$ be an observable.  Given $E \in \mathcal B_0(\mathbb R),$ there is an element $f_E \in \mathcal T$ such that $h(f_E)=x(E).$ Using (4.3), given $m \in \mathcal S_\sigma(M),$ there is a unique probability measure $P_m$ on $\mathcal B_0(\mathcal T)$ such that
$$
m(x(E))= m(h(f_E)) = \int_\Omega f_E(\omega)\, \dx P_m(\omega).
$$
We assert the latter integral does not depend on the choice of $f_E.$
Indeed, if $g_E$ is another function from $\mathcal T$ such that $h(g_E)=x(E),$ by (2) of Proposition \ref{pr:3.3}, the function $h_E:= \max\{f_E,g_E\}$ belongs to $\mathcal T$ and $h(h_E)=x(E).$ But $h_E-f_E \in \mathcal T$ and $h(h_E - f_E)=0$ so that
$$
0=m(h(h_E-f_E)) = \int_\Omega (h_E(\omega)-f_E(\omega))\, \dx P_m(\omega).
$$
In the similar way, we have
$$
0=m(h(h_E-g_E)) = \int_\Omega (h_E(\omega)-g_E(\omega))\, \dx P_m(\omega),
$$
whence
$$
\int_\Omega f_E(\omega)\, dP_m(\omega)= \int_\Omega g_E(\omega)\, \dx P_m(\omega).
$$

We assert that $P_m = m\circ h.$ Indeed,  let $f=\chi_A,$ $A\in \mathcal T.$ Then by (4.2),
$$
m(h(\chi_A))= \int_\Omega\chi_A(\omega)\, \dx P_m(\omega) =P_m(A).
$$

The mapping $\xi: \mathcal B_0(\mathcal T) \to M$ defined by $\xi(A):=h(\chi_A),$ $A \in \mathcal B_0(\mathcal T),$ is a sharp observable on $M.$
\end{proof}

Commenting Theorem \ref{th:5.3}, we say that the observable $x$ on $M$ is a {\it smearing} of a sharp observable $\xi$. This result extends an analogous result for $\sigma$-lattice effect algebras, see \cite[Thm 3.4]{JPV1}.

\begin{remark}\label{rem}
Let $(\Omega,\mathcal T,h)$ be a regular representation of a monotone $\sigma$-complete effect algebra with RDP such that all conditions of Theorem {\rm \ref{th:5.2}} are satisfied.  Then Theorem {\rm \ref{th:5.3}} holds also for our case of $(\Omega,\mathcal T, h).$
\end{remark}

\begin{proof}
It follows the same steps as the proof of Theorem \ref{th:5.3}.
\end{proof}

\begin{theorem}\label{th:5.4}
Let $M$ be a monotone $\sigma$-complete effect algebra with RDP and let $(\Omega,\mathcal T,h)$ be the canonical representation of $M$ such that every $f \in \mathcal T$ is $\mathcal B_0(\mathcal T)$-measurable. Given $a\in M$, there is a mapping $\Lambda_a: \mathcal B_0([0,1]) \to \mathcal \Sh(M)$ such that the mapping $a\mapsto \Lambda_a$ is injective, and for every $\sigma$-additive state $m$ on $M$, we have

$$
m(a)=\int_0^1 \lambda \, \dx m(\Lambda_a(\lambda)). \eqno(4.3)
$$

\end{theorem}

\begin{proof}
Let $(\Omega,\mathcal T,h)$ be the canonical representation of $M.$ Given $a \in M$,  choose a function $f=f_a \in \mathcal T$ such that $h(f)=a.$ For any Borel set $E \in \mathcal B_0([0,1]),$ let

$$
\Lambda_a(E):= h(\chi_{f_a^{-1}(E)}).\eqno(4.4)
$$
Then $\Lambda_a$ is by Theorem \ref{th:5.1} an observable on $\Sh(M).$

Assume that $g\in \mathcal T$ is another function such that $h(g)=a.$  As in the proof of Theorem \ref{th:5.3}, we can find a function $k \in \mathcal T$ such that $h(k)=a$ and $f,g \le k.$  Then $N(k-g)$ and $N(k-f)$ are meager sets, $k-f, k-g \in \mathcal T,$  and $\chi_{f^{-1}(E)} - \chi_{k^{-1}(E)} \in \mathcal T,$ $\chi_{g^{-1}(E)} - \chi_{k^{-1}(E)} \in \mathcal T$ for any $E=[0,t),$ $0\le t \le 1.$ In addition, $\chi_{f^{-1}(E)} - \chi_{k^{-1}(E)} \sim 0,$ $\chi_{g^{-1}(E)} - \chi_{k^{-1}(E)}\sim 0$ for any $E=[0,t),$ $0\le t \le 1.$ Therefore, $h(\chi_{f^{-1}(E)}) = h(\chi_{k^{-1}(E)})$ for every $E=[0,t).$  Let $\mathcal K$ be the set of Borel subsets $E$ from $[0,1]$ such that $h(\chi_{f^{-1}(E)}) = h(\chi_{k^{-1}(E)}).$ It is a Dynkin system, i.e. a system of subsets containing its universe  which is closed under the set theoretical complements and countable unions of of disjoint subsets.

The system $\mathcal K$ contains all intervals $(-\infty, t)\cap [0,1]$ for $t \in \mathbb R,$ all intervals of the form $[a,b)\cap [0,1],$ $a\le b$, as well as all finite unions of such  disjoint intervals $\bigcup_{i=1}^n [a_i,b_i)\cap [0,1].$  Because any finite union of intervals $\bigcup_{j=1}^m [c_j,d_j)\cap [0,1]$ can be expressed as a finite union of disjoint intervals, $\mathcal K$ contains also such unions. Therefore, if $E$ and $F$ are two finite unions of intervals, so is its intersection. Hence, by \cite[Thm 2.1.10]{Dvu}, $\mathcal K$ is also a $\sigma$-algebra, and finally we have $\mathcal K=\mathcal B_0([0,1]).$ In the same way, $h(\chi_{g^{-1}(E)}) = h(\chi_{k^{-1}(E)})$ and, consequently, $h(\chi_{f^{-1}(E)}) = h(\chi_{g^{-1}(E)})$ for every $E \in \mathcal B_0([0,1]).$

Consequently, we have proved that $\Lambda_a$ in (4.4) does not depend on the choice of $f.$

Now assume that $\Lambda_a = \Lambda_b$ for some $a,b \in M$ and let $h(f)=a$ and $h(g)=b$ for some $f,g \in \mathcal T.$ Then for every $E \in \mathcal B_0([0,1]),$ $h(\chi_{f^{-1}(E)}) = h(\chi_{g^{-1}(E)}).$

\begin{eqnarray*}
N(f-g )&=& \{s\in \mathcal S_\partial (M): f(s)\ne g(s)\}
= \bigcup_{r \in \mathbb Q}\{s\in \mathcal S_\partial (M): f(s)< r < g(s)\}\\
&=& \bigcup_{r \in \mathbb Q}(f^{-1}([0,r))\cap g^{-1}([r,1])) = \bigcup_{r \in \mathbb Q} f^{-1}([0,r))\Delta\, g^{-1}([0,r)),
\end{eqnarray*}
where $\Delta$ denotes the symmetric difference of two sets.
Since $h(\chi_{f^{-1}([0,r))}) = h(\chi_{g^{-1}([0,r))})$ for every rational $r \in [0,1],$ $N(\chi_{f^{-1}([0,r))} - \chi_{g^{-1}([0,r))})$ is a meager set. But $N(\chi_{f^{-1}([0,r))} - \chi_{g^{-1}([0,r))}) = f^{-1}([0,r))\Delta\, g^{-1}([0,r)),$ so that $N(f-g)$ is a meager set, and hence $a=h(f)=h(g)=b.$

Now let $m$ be an arbitrary $\sigma$-additive state on $M.$ It is clear that the mapping $m\circ h$ is a $\sigma$-additive state on $\mathcal T,$ and $m_h: \mathcal B_0(\mathcal T) \to [0,1],$ defined by $m_h(A)=m(h(\chi_A)),$ $A \in \mathcal B_0(\mathcal T),$ is a probability measure on $\mathcal B_0(\mathcal T).$ Given an element $a \in M$, there is an element $f_a \in \mathcal T$ such that $h(f_a)=a.$ Whence, the mapping $E \mapsto m(h(\chi_{f_a^{-1}(E)}),$ $ E\in \mathcal B_0([0,1]),$ is a probability measure on $\mathcal B_0([0,1]).$ By \cite[Thm 4.4]{Dvu3}, there is a unique probability measure $P_m$ on $\mathcal B_0(\mathcal T)$ such that
$$
m(h(f)) = \int_\Omega f(\omega) \, \dx P_m(\omega),\quad f \in \mathcal T.
$$
On the other hand, if $A \in \mathcal B_0(\mathcal T),$ then

$$
m(h(\chi_A)) =\int_\Omega \chi_A(\omega)\, \dx P_m(\omega) = \int_\Omega \chi_A(\omega) \, \dx m_h(\omega).
$$
Since every $f \in \mathcal T$ is $\mathcal B_0(\mathcal T)$-measurable, we have by \cite[Thm 4.4]{Dvu3}, that $P_m(A)=m_h(A)$ for every $A \in \mathcal B_0(\mathcal T).$

This yields for $f = f_a$

$$
m(h(f_a))= \int_\Omega f_a(\omega)\dx m_h(\omega)= \int_0^1 \lambda \, \dx m(h(\chi_{f_a^{-1}(\lambda)})) = \int_0^1 \lambda \, \dx m(\Lambda_a(\lambda)),
$$
and
$$
m(a) = \int_0^1 \lambda \, \dx m(\Lambda_a(\lambda))
$$
which proves (4.3).
\end{proof}

Theorem \ref{th:5.4} generalizes an analogous result from \cite{Pul}, and  the mapping $a \mapsto \Lambda_a$ is said to be the {\it spectral measure} of the element $a.$

\begin{remark}\label{rem1}
Let the conditions of Theorem {\rm \ref{th:5.4}} be satisfied.

\noindent {\rm (1)} If $a \in \Sh(M),$ then for $\Lambda_a$ defined by {\rm (4.4)}, we have
$$
\Lambda_a(E)=\left\{\begin{array}{lll}
a &\mbox{if}\ &   0\notin E,\ 1\in E,\\
a' &\mbox{if}\ &  0\in E, \ 1\notin E,\\
0 &\mbox{if}\   & 0,1 \notin E,\\
1 &\mbox{if}\   & 0,1 \in E,
\end{array}
\right.
\quad E \in \mathcal B_0([0,1]).
$$

{\rm (2)}  Let $\phi:[0,1]\to [0,1]$ be a strictly increasing and surjective Borel-measurable function such that $\phi(0)=0$ and $\phi(1)=1.$ If we define $\phi(\Lambda_a):\mathcal B_0([0,1])\to \Sh(M),$  $a \in M,$ by  $\phi(\Lambda_a)(E):=\Lambda_a(\phi^{-1}(E)),$ $E\in \mathcal B_0([0,1]),$ then the mapping  $a\mapsto \phi(\Lambda_a)$ is injective, but $\phi(\Lambda_a)$ is not necessarily a spectral measure because not always
$$
m(a)=\int_0^1 \phi(\lambda) \dx m(\Lambda_a(\lambda)).
$$

\end{remark}

\begin{proof} (1) Let $a \in \Sh(M).$ By Theorem \ref{th:5.1}, there is $A \in \mathcal B_0(\mathcal T)$ such that $h(\chi_A)=a.$ By (4.4), we have
$\Lambda_a(E)= h(\chi_{\chi_A^{-1}(E)}),$ for $E \in \mathcal B_0([0,1]).$ If $E =\{1\},$ then $\Lambda_a(\{1\}= h(\chi_A)=a.$ Similarly for other Borel sets $E.$

(2) Let $a \in M$ and let $h(f)=a.$  Using (4.4), we have $\phi(\Lambda_a)(E)= \Lambda_a(\phi^{-1}(E))=h(\chi_{f^{-1}(\phi^{-1}(E))}) = h(\chi_{(\phi\circ f)^{-1}(E)}). $ Assume now $\phi(\Lambda_a) = \phi(\Lambda_b)$ for some $b$ and let $h(g)=b.$ Then $h(\chi_{(\phi\circ f)^{-1}(E)}) = h(\chi_{(\phi\circ g)^{-1}(E)}), $ $E \in \mathcal B_0([0,1]).$
Similarly as in the proof of Theorem \ref{th:5.4}, we have that the set

\begin{eqnarray*}
N(f-g )&=& \{s\in \mathcal S_\partial (M): f(s)\ne g(s)\}\\
&=& \bigcup_{r \in \mathbb Q}\{s\in \mathcal S_\partial (M): f(s)< r < g(s)\}\\
&=& \bigcup_{r \in \mathbb Q}(f^{-1}(\phi^{-1}([0,\phi(r))))\cap g^{-1}(\phi^{-1}([\phi(r),1])))\\
& = &\bigcup_{r \in \mathbb Q} f^{-1}(\phi^{-1}([0,\phi(r))))\Delta\, g^{-1}(\phi^{-1}([0,\phi(r)))).
\end{eqnarray*}
Hence, we have $h(\chi_{f^{-1}(\phi^{-1}([0,\phi(r))))}) = h(\chi_{g^{-1}(\phi^{-1}([0,\phi(r))))})$  and $N(\chi_{f^{-1}(\phi^{-1}([0,\phi(r))))} - \chi_{g^{-1}(\phi^{-1}([0,\phi(r))))})$ is a meager set. But
 $$
N(\chi_{f^{-1}(\phi^{-1}([0,\phi(r))))} - \chi_{g^{-1}(\phi^{-1}([0,\phi(r))))}) = f^{-1}(\phi^{-1}([0,\phi(r))))\Delta\, g^{-1}(\phi^{-1}([0,\phi(r)))),
$$ so that $N(f-g)$ is a meager set, and hence $a=h(f)=h(g)=b.$

On the other hand, for any $\sigma$-additive state $m$ on $M,$ we have $  \int_0^1 \lambda \dx m(\phi(\Lambda_a(\lambda)))=\int_0^1 \phi(\lambda) \dx m(\Lambda_a(\lambda))$ which is not necessarily equal $m(a).$
\end{proof}

We note that we do not know whether the spectral measure is unique.

\begin{theorem}\label{th:5.5}
Let $M$ be a monotone $\sigma$-complete effect algebra with RDP and let $(\Omega,\mathcal T,h)$ be the canonical representation of $M$ such that every $f \in \mathcal T$ is $\mathcal B_0(\mathcal T)$-measurable.   Then every $\sigma$-additive state $m$ on the Boolean $\sigma$-algebra $\Sh(M)$ can be uniquely extended to a $\sigma$-additive state $\hat m$ on $M.$ In addition,
$$ \hat m(a)= \int_0^1 \lambda \dx m(\Lambda_a(\lambda)), \quad a \in M. \eqno(4.5)
$$
\end{theorem}

\begin{proof}
{\it Existence.} Let $m$ be a $\sigma$-additive state on $\Sh(M).$  Then the mapping $m_h(A):=m(h(\chi_A)),$ $A \in \mathcal B_0(\mathcal T),$ is due to Theorem \ref{th:5.1} a $\sigma$-additive measure on $\mathcal B_0(\mathcal T).$ Then
the mapping $s_m:\mathcal T \to [0,1]$ defined by
$$
s_m(f):= \int_\Omega f(\omega)\, \dx m_h(\omega), \quad f \in \mathcal T,
$$
is a $\sigma$-additive state on $\mathcal T.$ Now we define a function $\hat m: M \to [0,1]$ via $\hat m(a)=s_m(f)$ whenever $h(f)=a.$ We claim that $\hat m$ is defined correctly. Indeed, if $h(g)=a,$  by (2) of Proposition \ref{pr:3.3}, there is a function $k \in \mathcal T$ such that $h(k)=a$ and $f,g \le k.$  Therefore, $N(k-f)$ and $N(k-g)$ are meager sets, and $k-f, k-g \in \mathcal T.$ Then
$$
s_m(k-f)=\int_\Omega (k(\omega)-f(\omega))\dx m_h(\omega)=0= \int_\Omega (k(\omega)-g(\omega))\dx m_h(\omega),
$$
which entails $ \int_\Omega f(\omega)\dx m_h(\omega)= \int_\Omega g(\omega)\dx m_h(\omega).$

Assume $a+b$ is defined in $M,$ again by by (2) of Proposition \ref{pr:3.3}, we can assume that we have two functions $f,g\in \mathcal T$ such that $f\le 1-g,$ and $h(f)=a,$ $h(g)=b.$  Then $h(f+g)=a+b$ and
$$
\hat m(a+b):= \int_\Omega (f(\omega)+g(\omega)) \dx m_h(\omega)= \hat m(a)+ \hat m(b).
$$
It is clear that $\hat m$ is a state on $M.$  To show that $\hat m$ is $\sigma$-additive, assume $\{a_n\}\nearrow a$ in $M.$ Using by (2) of Proposition \ref{pr:3.3} and mathematical induction, we can assume that we have find a monotone sequence, $\{f_n\},$ of elements of $\mathcal T$ such that $h(f_n)=a_n$ and $h(a) = h(\bigvee_n f_n)= \bigvee_n a_n =a.$  Therefore,
$$
\hat m(a) = \int_\Omega f(\omega) \dx m_h(\omega) = \lim_n \int_\Omega f_n(\omega) \dx m_h(\omega) = \lim_n \hat m(a_n),
$$
which proves $\hat m$ is a $\sigma$-additive state on $M.$  Let now $a \in \Sh(M).$  By Theorem \ref{th:5.1}, there is $A \in \mathcal B_0(\mathcal T)$ such that $h(\chi_A)=a.$ Then
$$
\hat m(a) = \int_\Omega \chi_A \dx m_h(\omega)=m_h(A) = m(h(\chi_A))=m(a),
$$
which says that $\hat m$ is an extension of $m$ onto $M.$

\vspace{2mm} {\it Uniqueness.} Let $m_1$ and $m_2$ be two $\sigma$-additive extensions of $m$ onto the whole $M$.  Define $s_i(f):=m_i(h(f)),$ $f \in \mathcal T,$ $i=1,2.$  Then $s_i$ is a $\sigma$-additive state on $\mathcal T$ such that $s_1(\chi_A) =m(h(\chi_A))=s_2(\chi_A).$  By \cite[Thm 4.4]{Dvu3}, there are two probability measures $P_1$ and $P_2$ on $\mathcal B_0(\mathcal T)$ such that
$$
s_i(f)=\int_\Omega f(\omega) \dx P_i(\omega),\quad f \in \mathcal T.
$$
Then $s_1(\chi_A) =P_1(A)=m(h(\chi_A))=P_2(A) = s_2(\chi_A).$ Therefore, $s_1(f)=s_2(f)$ and if given $a\in M$, $h(f)=a$ for some $f \in \mathcal T,$ $m_1(a)=m_2(a).$

Finally, let $a \in M$ be given.  Choose $f \in \mathcal T$ such that $h(f)=a.$ Then using the Integral Transformation Theorem \cite{Hal}, we have

$$
m(a) =\int_\Omega f(\omega) \dx m_h(\omega)= \int_0^1 \lambda \dx m_h(f^{-1}(\lambda)) = \int_0^1 \lambda \dx m(\Lambda_a(\lambda)),
$$
when we have used (4.3).
\end{proof}

We recall that Theorem \ref{th:5.5} cannot be, in general, extended for any (finitely additive) state on $\Sh(M)$ of a monotone $\sigma$-complete effect algebra $M$ with RDP because this is possible iff $M$ is an MV-algebra as it was proved in \cite[Thm 5.1]{Dvu3}.

\end{document}